\documentclass[11pt,letter]{article}

\usepackage{amsfonts}
\usepackage{amsmath}
\usepackage{amsthm}
\usepackage{amssymb}
\usepackage{algorithm}
\usepackage[noend]{algpseudocode}
\usepackage[hyphens]{url}
\usepackage{hyperref}
\usepackage{graphicx}
\usepackage{subfigure}
\usepackage{xspace}
\usepackage{units}
\usepackage[margin=1in]{geometry}

\newcommand{\modelstyle}[1]{\textsc{#1}\xspace}
\newcommand{\fsynch}{\modelstyle{FSynch}}
\newcommand{\ssynch}{\modelstyle{SSynch}}
\newcommand{\asynch}{\modelstyle{ASynch}}

\newcommand{\look}{\modelstyle{Look}}
\newcommand{\compute}{\modelstyle{Compute}}
\newcommand{\move}{\modelstyle{Move}}
\newcommand{\rend}{\emph{Rendezvous}\xspace}

\newcommand{\FS}{\modelstyle{FState}}
\newcommand{\FC}{\modelstyle{FComm}}

\newcommand{\oneh}{\nicefrac{1}{2}\xspace}
\newcommand{\oneq}{\nicefrac{1}{4}\xspace}
\newcommand{\oneo}{\nicefrac{1}{8}\xspace}

\newtheorem{theorem}{Theorem}

\newtheorem{lemma}[theorem]{Lemma}
\newtheorem{observation}[theorem]{Observation}

\begin{document}
\thispagestyle{empty}

\title{\bf Rendezvous of Two Robots with Constant Memory}

\author{
P. Flocchini\footnotemark[1]
\and
N. Santoro\footnotemark[2]
\and
G. Viglietta\footnotemark[2]
\and
M. Yamashita\footnotemark[3]
}

\def\thefootnote{\fnsymbol{footnote}}
\footnotetext[1]{School of Electrical Engineering and Computer Science, University of Ottawa, Canada.
}
\footnotetext[2]{School of Computer Science, Carleton University, Canada.
}
\footnotetext[3]{Kyushu University, Fukuoka, Japan.
}
 
\date{}

\maketitle

\begin{abstract}
We study the impact that persistent memory has on the
classical {\em rendezvous} problem of two mobile computational entities, called robots, in the plane. It is well known that,
without additional assumptions,
rendezvous is impossible if the entities are oblivious (i.e., have no
persistent memory) even if the system is semi-synchronous ({\ssynch}). It has been recently shown that
 rendezvous is possible even if the
system is asynchronous 
 ({\asynch})
if each robot is endowed with $O(1)$ bits of persistent  memory, can  transmit  $O(1)$ bits in each cycle, and can remember
(i.e., can persistently store) the last  received transmission. This setting 
 is overly powerful. 

In this paper we weaken that setting in two different ways: (1)  by maintaining the
$O(1)$ bits of persistent  memory but removing the communication capabilities;
and (2) 
by maintaining the  $O(1)$ transmission capability  and  the ability to  remember
 the last  received transmission, but 
 removing the ability of an agent to remember its previous activities.
 We call the former setting   {\em finite-state}  (\FS)
 and  the latter   {\em finite-communication} (\FC).
 Note that, even though its use is very different,  in both settings, the amount of persistent memory of a robot is constant.
 
 We investigate the rendezvous problem in these two  weaker settings. We model both settings as a system of robots
 endowed with visible lights: in \FS, a robot can only see its own light, while in \FC a robot can only see the other robot's light.
We prove, among other things,  that 
finite-state  robots can rendezvous in \ssynch, and that   finite-communication robots are able to rendezvous
even in {\asynch}. All proofs are constructive: in each setting, we present a protocol that allows the two robots to
rendezvous in finite time.

\end{abstract}

\section{Introduction}

\subsection{Framework and Background}

{\em Rendezvous}  is 
the process of two computational mobile entities, initially dispersed in a spatial universe, meeting within
finite time at a location, non known \emph{a priori}. When there are more than two entities, this task is known as {\em Gathering}.
These two problems are   core problems in distributed computing by mobile entities. They
have been  intensively and extensively studied  when the universe
 is a connected region of ${\mathbb R}^2$ 
 in which the entities, usually called {\em robots}, can freely move;
 see, for example,  \cite{agmon,BoGT09a,gather,CoPe04,Cord2011,degener2011,DiPe12,FloPSW05,KLOT11c,LinMA07a,viglietta2012,Pr07,SouDY09}.
 
 Each entity is modeled as a point, it  has its own local coordinate system of which it perceives itself as the centre, and has its own unit distance. Each  entity  operates in cycles of   \look, \compute,   \move  activities.
 In each cycle, an entity observes the position of the other entities expressed in its local coordinate system (\look);  using that
 observation as input, it executes a protocol (the same for all robots) and computes a destination point  (\compute); it then moves   to the computed destination point (\move).
 Depending on the activation schedule and the synchronization level, three
 basic types of systems are identified in the literature:  a {\em fully synchronous} system (\fsynch) is equivalent to
 a system where  there is a common clock and at each clock tick all entities are activated simultaneously,
 and \compute and \move are instantaneous;   a  {\em semi-synchronous} system (\ssynch) is like a fully synchronous one except that,
 at each clock tick,  only some entities will be activated (the choice is made by a fair scheduler); in
 a {\em fully asynchronous} system (\asynch), there is no common notion of time, each \compute and \move of each robot can take an unpredictable (but finite) amount of time, and the interval of time between successive activities is finite but unpredictable.
The focus of almost all algorithmic investigations in the continuous setting has been  on
 {\em oblivious} robots, that is when the  memory of the robots is erased at the end of each cycle,
 in other words the robots have no persistent memory (e.g., for an overview see  \cite{book}).

 The importance of  {\em Rendezvous} in the continuous setting derives in part from the fact that 
it separates \fsynch from \ssynch for oblivious robots. Indeed,
{\em Rendezvous} is trivially solvable in a fully synchronous system, without any additional assumption.
However, without additional assumptions,
{\em Rendezvous} is {\em impossible} for oblivious robots  if the system is semi-synchronous \cite{ssynch}.

Interestingly, from a computational point of view, {\em Rendezvous}  is very different from the  {\em Gathering}  problem  of having 
$k\geqslant 3$   robots meet
 in the same point; in fact,  {\em Gathering}  of oblivious robots is always {\em possible}
for any
$k\geqslant 3$   
  even in \asynch   without any additional assumption other than multiplicity detection \cite{gather}.  Furthermore,
  in  \ssynch,  $k\geqslant 3$  robots can  gather  even in spite of a certain number of faults \cite{agmon,BoDT13,DefGMR06},
  and converge in spite of inaccurate measurements \cite{cohen06}; see also \cite{IzBTW11}.
  
The {\em Rendezvous} problem also shows the impact of certain factors.
For example, the problem has a trivial solution if the robots are endowed with {\em consistent compasses}
even if the system is fully asynchronous.
The problem is  solvable in \asynch even  if the local compasses have some degree of inconsistency
(a tilt of an appropriate angle) \cite{ISKIDWY}; the solution is no longer trivial, but does exist.

In this paper, we are concerned with the impact that {\em memory} has on the solvability of the {\em Rendezvous} problem.
In particular, we are interested in determining what type and how much persistent memory would allow
the robots to rendezvous.  What is known in this regard is very little. On the one hand, it is well known that, in absence of additional
assumptions,  without persistent memory  rendezvous is impossible even in \ssynch \cite{ssynch}.
On the other hand, a recent result shows that
  rendezvous is indeed possible even in \asynch if each robot has $O(1)$ bits of persistent  memory {\em and}
 can  transmit  $O(1)$ bits in each cycle  {\em and} can remember
(i.e., can persistently store) the last  received transmission
 \cite{visbits} (see also \cite{viglietta} for size-optimal solutions). 

The conditions of the latter result are  overly powerful.
The natural question is whether the simultaneous presence of
these conditions is truly necessary for rendezvous.

\subsection{Main Contributions}

In this paper we  address this question by  weakening the setting in two different ways, and investigate the {\em Rendezvous} problem in these weaker settings. Even though its use is very different,  in both settings, the amount of persistent memory of a robot is constant.

We  first  examine the      setting  where the two robots have  $O(1)$ bits of
{\em internal}  persistent memory but  cannot communicate; this corresponds to the  {\em   finite-state} (\FS) robots model.
Among other contributions, we prove that
\FS  robots with rigid movements can rendezvous in \ssynch, and that this can be done
using only six internal states. 
The proof is constructive:   we present a protocol that allows the two robots to
rendezvous in finite time under the stated conditions.

We then study the    {\em  finite-communication}   (\FC) setting,
where a robot   can  transmit  $O(1)$ bits in each cycle and  remembers 
the last  received transmission,  but it is otherwise oblivious:  it has no other persistent memory
of its previous observations, computations and transmissions. 
We prove that  two \FC robots with rigid movements
 are able to rendezvous
even in {\asynch};  this is doable even if the different messages that can be sent are just 12. We also prove that only three different messages suffice in \ssynch.
Also for this model  all the proofs are constructive.

Finally,  we consider the {\em Rendezvous} problem when the movement of the robots can be 
interrupted by an adversary (in the above results,  in each cycle  a robot  reaches 
its computed destination point).
 The only constraint on the adversary is that a robot  moves at least a distance $\delta>0$
(otherwise, rendezvous is clearly impossible).
We show that,  with  knowledge of $\delta$,    three internal states are sufficient to solve {\em Rendezvous}  by
{\FS  robots  in \ssynch,
and  three possible messages are sufficient for
\FC robots in  \asynch.
In other words, we prove that rigidity of the movements can be traded with knowledge of $\delta$.

These results are obtained  modeling both settings  as a system of robots
 endowed with a constant number of \emph{visible lights}: a \FS robot can  see only its own light, while a \FC  robot can  see only the other robot's light. Our  results seem to indicate that ``it is better to communicate than to remember''.

In addition to the  specific results on the {\em Rendezvous} problem, an important contribution of this paper
is the  extension of  the classical model of oblivious silent robots into
two  directions: adding finite memory, and enabling finite communication. 

\section{Model and Terminology}
The general model we employ is the standard one,  described in~\cite{book}. The two robots  are autonomous  computational entities
 modeled as points  moving in $\mathbb R^2$. Each robot has its own coordinate system and its own unit distance, which may differ from each other,
 and it always perceives itself as lying at the origin of its own local coordinate system. Each robot operates in cycles that consist of three phases:  \look, \compute, and \move.
In the \look phase it gets the position (in its local coordinate system) of the other robot; in the \compute phase,  it computes a destination point; in the  \move phase it moves to the computed destination point, along a straight line. Without loss of generality, the \look phase is assumed to be instantaneous. 
The robots are anonymous and oblivious, meaning that they do not have distinct identities, they  execute the same algorithm in each \compute phase, and the  input to such algorithm is the snapshot coming from the previous \look phase.

 Here we study two settings; both settings can be described as restrictions of the model
  of  visibile lights introduced in \cite{visbits}. In that model,
 each robot  carries  a persistent memory of constant size, called  {\em light}; the value of the light is called {\em color} or {\em state}, and it is set
by the robot during each \compute phase. Other than their own light, the robots have no other persistent memory of past snapshots and computations.

In the first setting, that of {silent finite-state} (or simply, \FS) robots,   the light of a robot is  visible  only to the robot itself; i.e., the colored light merely encodes an \emph{internal state}.  In the second setting, of  {oblivious finite-communication} (or simply \FC)  robots,  the light of a robot is  visible  only to the other robot; i.e., 
 they can communicate with the other robot through their colored light,  but by their next cycle they forget even
 the color of their own light  (since they do not see  it).
The color a robot sees is used as input during the computation.

 In the {\em asynchronous} (\asynch) model,
the robots are activated independently,
and the duration of each \compute, \move and  inactivity 
is finite but unpredictable. As a result, the robots do not have a common notion
of time, robots can be seen while moving, and computations can be made based
on obsolete observations.
In the   {\em semi-synchronous}
(\ssynch) models
the activations of  robots can be logically divided into global
rounds; in each round, one or both robots are  activated, obtain  the same snapshot, compute, and perform
their move.
It is assumed that the activation schedule  is fair, i.e., each robot is activated infinitely often.

Depending on whether or not the adversary can stop a robot before it reaches its computed destination, the movements are called {\em non-rigid}
and {\em rigid}, respectively. In the case of non-rigid movements,     there exists a constant $\delta > 0$ 
 such that if the destination point's distance
is smaller than $\delta$,  the robot will reach it; otherwise, it  will move
towards it by at least $\delta$.
Note that, without this assumption, an adversary could make it impossible for any robot to ever reach its destination, following a classical Zenonian argument. 

The two robots solve the {\em Rendezvous} problem if, within finite time, they move to the same point  and do not move from there;
the  meeting point is not determined \emph{a priori}.
A rendezvous algorithm for \ssynch (resp., \asynch) is a protocol that allows the robots to solve the {\em Rendezvous} problem under any possible 
execution schedule in \ssynch  (resp., \asynch).
A particular class of algorithms,   denoted by   $\mathcal L$, is that where
each robot may only compute 
a destination point of the form $\lambda\cdot other.position$, for some  $\lambda\in \mathbb R$
 obtained as a function  only of the light  of which the robot is aware (i.e., its internal state in the \FS model, or the other robot's color in the \FC model).
The algorithms of this class are of interest because they
operate also when the coordinate system of a robot is not  self-consistent (i.e., it can unpredictably rotate, change its scale or undergo a reflection).

\section{Finite-State Robots}
We fist consider \FS robots and we   start by identifying a simple impossibility result 
for algorithms  in $\mathcal L$.

\begin{theorem}\label{th:FSImpo}
In  \ssynch, \rend of two  \FS robots   is unsolvable by algorithms in $\mathcal L$, regardless of the amount of their internal memory.
\end{theorem}
\begin{proof}
For each robot, the destination point and the next state are a function of the internal state only. Assuming that both robots start in the same state, we keep activating them one at a time, alternately. Hence, every other turn they are in the same state. As soon as the first robot attempts to move to the other robot's location, we activate both robots simultaneously, making them switch positions. By repeating this pattern, the robots never gather.
\end{proof}

\noindent Thus the computation of the destination 
must  take into account more than just the
 lights (or states) of which the robot is aware. 
 
 The approach we use to circumvent this impossibility result is to have each robot
 use its own unit of distance as a computational tool; recall that the two robots might have different units, and they are not known to each other.
We propose Algorithm~\ref{alg1} for \rend in   \ssynch, also illustrated in Figure~\ref{f1}. 
Each robot has six internal states, namely $S_{\rm start}$, $S_1$, $S_2^{\rm left}$, $S_2^{\rm right}$, $S_3$, and $S_{\rm finish}$. Both robots are assumed to begin their execution in $S_{\rm start}$. Each robot lies in the origin of its own local coordinate system
 and the two robots have no agreement on axes orientations or unit distance.

Intuitively, the robots try to reach a configuration in which they both observe the other robot at distance not lower than $1$ (their own unit). From this configuration, they attempt to meet in the midpoint. If they never meet because they are never activated simultaneously, at some point one of them notices that its observed distance is lower than $1$. This implies a breakdown of symmetry that enables the robots to finally gather.

In order to reach the desired configuration in which they both observe a distance not lower than $1$, the two robots first try to move farther away from each other if they are too close. If they are far enough, they memorize the side on which they see each other (left or right), and try to switch positions. If only one of them is activated, they gather; otherwise they detect a side switch and they can finally apply the above protocol.
This is complicated by the fact that the robots may disagree on the distances they observe. To overcome this difficulty, they use their ability to detect a side switch to understand which distance their partner observed. If the desired configuration is not reached because of a disagreement, a breakdown of symmetry occurs, which is immediately exploited to gather anyway.
As soon as the two robots coincide at the end of a cycle, they never move again, and \rend is solved.

\begin{algorithm}
\caption{\ Rendezvous for rigid \ssynch with no unit distance agreement and six internal states}\label{alg1}
\begin{algorithmic}[1]
\State $dist \gets \lVert other.position\rVert$
\If{$dist = 0$}
	\State terminate
\EndIf
\If{$other.position.x > 0$}
	\State $dir \gets$ right
\ElsIf{$other.position.x < 0$}
	\State $dir \gets$ left
\ElsIf{$other.position.y > 0$} \Comment{$other.position.x = 0$}
	\State $dir \gets$ right
\Else
	\State $dir \gets$ left
\EndIf
\If{$me.state = S_{\rm start}$}
	\If{$dist < 1$}
		\State $me.state \gets S_1$
		\State $me.destination \gets other.position\cdot (1-1/dist)$ 
	\Else
		\State $me.state \gets S_2^{dir}$
		\State $me.destination \gets other.position$
	\EndIf
\ElsIf{$me.state = S_1$}
	\If{$dist \leqslant 1$}
		\State $me.state \gets S_{\rm finish}$
		\State $me.destination \gets (0,0)$
	\Else
		\State $me.state \gets S_2^{dir}$
		\State $me.destination \gets other.position$
	\EndIf
\ElsIf{$me.state = S_2^d$}
	\If{$dir = d$}
		\State $me.state \gets S_{\rm finish}$
		\State $me.destination \gets other.position$
	\ElsIf{$dist < \oneh$} \Comment{side switch detected}
		\State $me.state \gets S_{\rm finish}$
		\State $me.destination \gets (0,0)$
	\Else
		\State $me.destination \gets other.position/2$
		\If{$dist < 1$}
			\State $me.state \gets S_3$
		\EndIf
	\EndIf
\algstore{myalg1}
\end{algorithmic}
\end{algorithm}

\begin{algorithm}
\begin{algorithmic}[1]
\algrestore{myalg1}
\ElsIf{$me.state = S_3$}
	\State $me.state \gets S_{\rm finish}$
	\If{$dist < \oneq$}
		\State $me.destination \gets (0,0)$
	\Else \Comment{$\oneq \leqslant d < \oneh$}
		\State $me.destination \gets other.position$
	\EndIf
\Else \Comment{$me.state = S_{\rm finish}$}
	\If{$dist \leqslant 1$}
		\State $me.destination \gets (0,0)$
	\Else
		\State $me.destination \gets other.position$
	\EndIf
\EndIf
\end{algorithmic}
\end{algorithm}

\begin{figure}[ht]
\centering{\includegraphics[width=10.2cm]{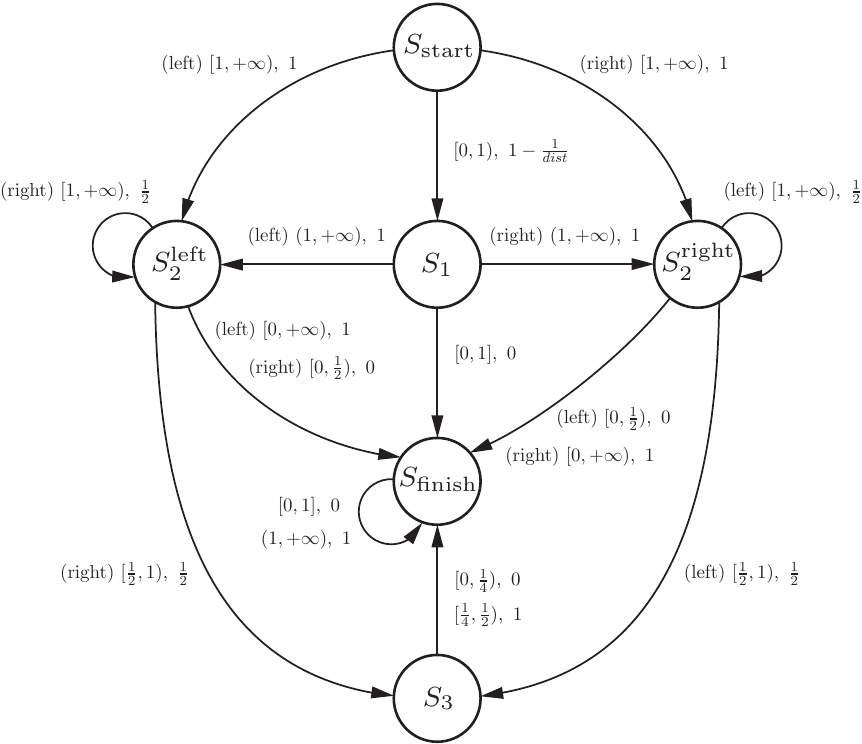}}
\caption{Illustration of Algorithm~\ref{alg1}. A label of the form $(d) I, \lambda$ denotes a transition that applies when the other robot is seen in direction $d\in \{{\rm left}, {\rm right}\}$ and its observed distance lies in the interval $I\subset \mathbb{R}$. The computed destination point is $\lambda \cdot other.position$. For example, a robot in state $S_{\rm start}$ perceiving the other at distance $\geqslant 1$ on the right  will move to the position of the other robot and will change state to $S_2^{\rm right}$.
}
\label{f1}
\end{figure}

To analyze the correctness  of Algorithm~\ref{alg1}, some terminology is needed.  
In the following, the two robots will be called $r$ and $s$, respectively. An expression of the form $(S_r, S_s, I_r, I_s)$ denotes a configuration in which robot $r$ (resp.\ $s$) is in state $S_r$ (resp.\ $S_s$), and the distance at which it sees the other robot lies in the interval $I_r$ (resp.\ $I_s$), according to its own distance function. Therefore, the starting configuration of $r$ and $s$ is $(S_{\rm start}, S_{\rm start},$ $[0,+\infty), [0,+\infty))$.

With abuse of notation, we will say that a robot is in state $S_2^{=}$ if it is in state $S_2^{\rm left}$ (resp.\ $S_2^{\rm right}$) and it sees the other robot on its left (resp.\ right). Analogously, a robot is said to be in state $S_2^{\neq}$ if its state is $S_2^{\rm left}$ or $S_2^{\rm right}$ and it has detected a switch.
For a configuration $C$, the expression $C\downarrow$ means that, whenever the two robots reach $C$, \rend is eventually solved.

\begin{observation}\label{o2}
$(S_a, S_b, I_a, I_b) \downarrow$ if and only if $(S_b, S_a, I_b, I_a) \downarrow$.
\end{observation}

\begin{observation}
If $(S_r, S_s, I_r, I_s) \downarrow$ and $I'_r\subseteq I_r$, then $(S_r, S_s, I'_r, I_s) \downarrow$.
\end{observation}

\begin{lemma}\label{l4}
$(S_{\rm finish}, S_{3}, [0,1], [\oneq,\oneh)) \downarrow$.
\end{lemma}
\begin{proof}
Robot $r$ keeps staying still, while robot $s$ moves to $r$ as soon as it is activated.
\end{proof}

\begin{lemma}\label{l5}
$(S_{\rm finish}, S_2^{\neq}, [0,1], [\oneh,+\infty)) \downarrow$.
\end{lemma}
\begin{proof}
Robot $s$ keeps moving to the midpoint, while robot $r$ never moves, because it keeps observing a distance not greater than $1$. As soon as $s$ observes a distance smaller than $1$ (hence in $[\oneh, 1)$), its state becomes $S_3$ and it moves to the midpoint again. Now the configuration is $(S_{\rm finish}, S_{3}, [0,\oneh], [\oneq,\oneh))$, and Lemma~\ref{l4} applies.
\end{proof}

\begin{lemma}\label{l9}
$(S_{\rm finish}, S_2^{=}, [0,1], [0,+\infty)) \downarrow$.
\end{lemma}
\begin{proof}
Robot $r$ keeps staying still, while robot $s$ moves to $r$ as soon as it is activated.
\end{proof}

\begin{lemma}\label{l7}
$(S_3, S_2^{\neq}, [\oneq,\oneh), [0,+\infty)) \downarrow$.
\end{lemma}
\begin{proof}
If only robot $r$ is activated, it moves to $s$, and \rend is solved.

If only robot $s$ is activated, two cases arise. If the distance observed by $s$ is less than $\oneh$, configuration $(S_3, S_{\rm finish}, [\oneq,\oneh), [0,\oneh))$ is reached, and Lemma~\ref{l4} applies. Otherwise, if the distance observed by $s$ is at least $\oneh$, $s$ moves to the midpoint (and possibly switches to $S_3$). As a consequence, the distance observed by $r$ becomes less than $\oneq$, hence it stays still forever (it only switches to $S_{\rm finish}$ as soon as it is activated). On the other hand, $s$ keeps moving to the midpoint, until it observes a distance lower than $1$, switches to $S_3$, and finally moves to $r$, solving \rend.

If both robots are activated on the first cycle, three cases arise.
\begin{itemize}
\item If the distance observed by $s$ is less than $\oneh$, $r$ moves to $s$ and $s$ stays still, hence \rend is solved.
\item If the distance observed by $s$ lies in $[\oneh,1)$, configuration $(S_{\rm finish}, S_3,$ $[\oneo,\oneq), [\oneq,\oneh))$ is reached, and Lemma~\ref{l4} applies.
\item If the distance observed by $s$ is at least $1$, configuration $(S_{\rm finish}, S_2^{=},$ $[\oneo,\oneq), [\oneh,+\infty))$ is reached (the two robots switch sides), and Lemma~\ref{l9} applies.
\end{itemize}
\end{proof}

\begin{lemma}\label{l8}
$(S_2^{\neq}, S_2^{\neq}, [0,\oneh), [\oneh,+\infty)) \downarrow$.
\end{lemma}
\begin{proof}
If both robots are activated, two cases arise.  If the distance observed by $s$ is less than $1$, configuration $(S_{\rm finish}, S_3, [0,\oneq), [\oneq,\oneh))$ is reached, and Lemma~\ref{l4} applies. Otherwise, if the distance is at least $1$, configuration $(S_{\rm finish}, S_2^{\neq}, [0,\oneq), [\oneh,+\infty))$ is reached, and Lemma~\ref{l5} applies.

If only $r$ is activated, configuration $(S_{\rm finish}, S_2^{\neq}, [0,\oneh), [\oneh,+\infty))$ is reached, and Lemma~\ref{l5} applies.

If only robot $s$ is activated, two cases arise. If the distance observed by $s$ is less than $1$, configuration $(S_2^{\neq}, S_3, [0,\oneq), [\oneq,\oneh))$ is reached, and Lemma~\ref{l7} applies. Otherwise, if the distance is at least $1$, the configuration remains $(S_2^{\neq}, S_2^{\neq}, [0,\oneh), [\oneh,+\infty))$, but the distance between the two robots has halved. As the execution progresses, this case cannot repeat forever, because eventually the distance observed by $s$ becomes less than $1$, or $r$ is activated.
\end{proof}

\begin{lemma}\label{l2}
$(S_2^{\neq}, S_2^{\neq}, [1,+\infty), [1,+\infty)) \downarrow$.
\end{lemma}
\begin{proof}
If both robots are activated, they compute the midpoint and they gather. If only one robot is activated at each cycle, configuration $(S_2^{\neq}, S_2^{\neq},$ $[1,+\infty), [1,+\infty))$ keeps repeating for finitely many cycles, until the distance observed by some robot, say $r$, becomes less than $1$. The configuration then becomes $(S_2^{\neq}, S_2^{\neq}, [\oneh,1), [\oneh,+\infty))$.

Once again, if both robots are activated at the next cycle, they gather in the midpoint. If only $r$ is activated, configuration $(S_3, S_2^{\neq}, [\oneq,\oneh), [\oneq,+\infty))$ is reached, and Lemma~\ref{l7} applies. On the other hand, if only $s$ is activated, two cases arise. If the distance observed by $s$ is less than $1$, configuration $(S_2^{\neq}, S_3, [\oneq,\oneh), [\oneq,\oneh))$ is reached, and Lemma~\ref{l7} applies again. If the distance is at least $1$, then configuration $(S_2^{\neq}, S_2^{\neq}, [\oneq,\oneh), [\oneh,+\infty))$ is reached, and Lemma~\ref{l8} applies.
\end{proof}

\begin{lemma}\label{l6}
$(S_1, S_{\rm finish}, [0,1], (1,+\infty)) \downarrow$ and $(S_1, S_{\rm start}, [0,1], [1,+\infty)) \downarrow$.
\end{lemma}
\begin{proof}
Robot $r$ switches to $S_{\rm finish}$ as soon as it is activated, and keeps staying still. Robot $s$ moves to $r$ as soon as it is activated.
\end{proof}

\begin{lemma}\label{l3}
$(S_1, S_{\rm start}, \{1\}, [0,+\infty)) \downarrow$.
\end{lemma}
\begin{proof}
If the distance observed by robot $s$ is at least $1$, then Lemma~\ref{l6} applies. Otherwise the configuration is $(S_1, S_{\rm start}, \{1\}, [0,1))$. Three cases arise.
\begin{itemize}
\item If both robots are activated at the first cycle, they reach configuration $(S_{\rm finish}, S_1, (1,+\infty), \{1\})$, and Lemma~\ref{l6} applies.
\item If only robot $r$ is activated at the first cycle, configuration $(S_{\rm finish}, S_{\rm start},$ $\{1\}, [0,1))$ is reached. Now $r$ keeps staying still and in state $S_{\rm finish}$. As soon as $s$ is activated, configuration $(S_{\rm finish}, S_1, (1,+\infty), \{1\})$ is reached, and Lemma~\ref{l6} applies.
\item If only robot $s$ is activated at the first cycle, configuration $(S_1, S_1,$ $(1,+\infty), \{1\})$ is reached. From now on, $s$ keeps staying still (possibly switching to $S_{\rm finish}$), whereas $r$ moves to $s$ as soon as it is activated.
\end{itemize}
\end{proof}

\begin{lemma}\label{l1}
$(S_1, S_1, (1,+\infty), (1,+\infty)) \downarrow$.
\end{lemma}
\begin{proof}
If only one robot is activated, it moves to the other robot, and \rend is solved. If both robots are activated, they turn to $S_2^{\rm left}$ or $S_2^{\rm right}$ and switch positions. Hence they reach configuration $(S_2^{\neq}, S_2^{\neq}, (1,+\infty),$ $(1,+\infty))$, and Lemma~\ref{l2} applies.
\end{proof}

\begin{theorem}\label{t1}
In \ssynch, \rend of two \FS robots is solvable with six internal states. This result holds even without unit distance agreement.
\end{theorem}
\begin{proof}
We prove that $(S_{\rm start}, S_{\rm start}, [0,+\infty), [0,+\infty)) \downarrow$. Three cases arise.
\begin{itemize}
\item Let the configuration be $(S_{\rm start}, S_{\rm start}, [0,1), [0,1))$. If only one robot is activated, say $r$, then configuration $(S_1, S_{\rm start}, \{1\}, [0,+\infty))$ is reached, and Lemma~\ref{l3} applies. If both robots are activated, configuration $(S_1, S_1, (1,+\infty), (1,+\infty))$ is reached, and Lemma~\ref{l1} applies.
\item Let the configuration be $(S_{\rm start}, S_{\rm start}, [1,+\infty), [0,1))$ (the symmetric case is equivalent, due to Observation~\ref{o2}). If only robot $r$ is activated, it moves to $s$ and \rend is solved. If only $s$ is activated, configuration $(S_{\rm start}, S_1, (1,+\infty), \{1\})$ is reached, and Lemma~\ref{l6} applies. Finally, if both robots are activated, configuration $(S_2^{=}, S_1, [0,+\infty), [0,1))$ is reached. Next, if only robot $s$ is activated, configuration $(S_2^{=}, S_{\rm finish},$ $[0,+\infty), [0,1))$ is reached, and Lemma~\ref{l9} applies. In any other case, $r$ moves to $s$ and \rend is solved.
\item Let the configuration be $(S_{\rm start}, S_{\rm start}, [1,+\infty), [1,+\infty))$. If only one robot is activated, it moves to the other robot, and \rend is solved. If both robots move, they switch positions, and the configuration becomes $(S_2^{\neq}, S_2^{\neq}, [1,+\infty), [1,+\infty))$. Then Lemma~\ref{l2} applies.
\end{itemize}
\end{proof}

\section{Finite-Communication Robots}
We now focus on \FC robots distinguishing the asynchronous and the semi-synchronous cases.

\subsection{Asynchronous}
It is not difficult  to see that algorithms in $\mathcal L$ are not sufficient to solve the problem.

\begin{theorem}\label{th:externalImp}
 In \asynch, \rend of two \FC robots is unsolvable  by algorithms in $\mathcal L$, regardless of the amount of colors employed.
\end{theorem}
\begin{proof}
For each robot, the destination point and the next state are a function of the state of the other robot only. Assuming that both robots start in the same state, we let them perform their execution synchronously. As soon as both robots compute the midpoint $m$ as a result of seeing each other in state $A$, we let only robot $r$ complete its cycle. Meanwhile, $s$ has computed $m$ but still has not updated its state, nor moved. Therefore, $r$ keeps seeing $s$ set to $A$, and computes the new midpoint without changing its own state. We let $r$ complete another cycle, and then we let $s$ update its state and reach $m$. As a result, both robots are back in the same state and have not gathered. By repeating this pattern, the robots never solve \rend.
\end{proof}

We now describe an  algorithm  (which is not  in $\mathcal L$) that solves the problem.
Also this algorithm uses
 the local unit distance as a computational tool, but in a rather different way since a robot  cannot remember and  has to infer information by observing the other robot's light.
 
Intuitively, the two robots try to reach a configuration in which both robots
  see each other at distance lower than $1$. To do so, they first
communicate to the  other whether or not the distance they observe is
smaller than $1$ (recall that they may disagree, because their unit
distances may differ). If one robot acknowledges that its partner has
observed a distance not smaller  than $1$, it reduces the distance by
moving toward the midpoint.

The process goes on until both robots observe a distance smaller than
$1$. At this point, if they have not gathered yet, they try to compare
their distance functions, in order to break symmetry. They move away
from each other in such a way that their final distance is the sum of
their respective unit distances.
Before proceeding, they attempt to switch positions. If, due to
asynchrony, they failed to be in the same state at any time before
this step, they end up gathering. Instead, if their execution has been
synchronous up to this point, they finally switch positions.
Now, if the robots have not gathered yet, they know that their
distance is actually the sum of their unit distances. Because each
robot knows its own unit, they can tell if one of them is larger. If a
robot has a smaller unit, it moves toward its partner, which waits.

Otherwise, if their units are equal, they apply a simple protocol: as
soon as a robot wakes up, it moves toward the midpoint and orders its
partner to stay still. If both robots do so, they gather in the
middle. If one robot is delayed due to asynchrony, it acknowledges the
order to stay still and tells the other robot to come.

\begin{algorithm}
\caption{\ Rendezvous for rigid \asynch with no unit distance agreement and 12 externally visible states}\label{alg3x}
\begin{algorithmic}[1]
\State $dist \gets \lVert other.position\rVert$
\If{$other.state =$ (\textsc{Test})} \Comment{testing distances}
	\If{$dist \geqslant 1$}
		\State $me.state \gets$ (\textsc{Me} $\geqslant 1$)
	\Else
		\State $me.state \gets$ (\textsc{Me} $<1$)
	\EndIf
\ElsIf{$other.state =$ (\textsc{Me} $\geqslant 1$)} \Comment{reducing distances}
	\State $me.state \gets$ (\textsc{Approaching})
	\State $me.destination \gets other.position/2$
\ElsIf{$other.state =$ (\textsc{Approaching})} \Comment{test distances again}
	\State $me.state \gets$ (\textsc{Test})
\ElsIf{$other.state =$ (\textsc{Me} $<1$)}
	\If{$dist \geqslant 1$}
		\State $me.state \gets$ (\textsc{Me} $\geqslant 1$)
	\Else
		\State $me.state \gets$ (\textsc{Both} $<1$)
	\EndIf
\ElsIf{$other.state =$ (\textsc{Both} $<1$)}
	\If{$dist = 0$} \Comment{we have gathered}
		\State $me.state \gets$ (\textsc{Halted})
	\Else
		\State $me.state \gets$ (\textsc{Moving Away})
		\If{$dist <1$} \Comment{moving away by $1-dist/2$}
			\State $me.destination \gets other.position\cdot (\oneh-1/dist)$
		\EndIf
	\EndIf
\ElsIf{$other.state =$ (\textsc{Moving Away})}
	\State $me.state \gets$ (\textsc{You Moved})
\ElsIf{$other.state =$ (\textsc{You Moved})}
	\State $me.state \gets$ (\textsc{Coming})
	\State $me.destination \gets other.position$
\ElsIf{$other.state =$ (\textsc{Coming})}
	\State $me.state \gets$ (\textsc{Waiting})
\ElsIf{$other.state =$ (\textsc{Waiting})}
	\If{$dist > 2$} \Comment{my unit is smaller}
		\State $me.state \gets$ (\textsc{Stay})
		\State $me.destination \gets other.position$
	\ElsIf{$dist = 2$} \Comment{our units are equal}
		\State $me.state \gets$ (\textsc{Both} $=2$)
	\Else \Comment{my unit is bigger or we have gathered}
		\State $me.state \gets$ (\textsc{Halted})
	\EndIf
\ElsIf{$other.state =$ (\textsc{Both} $=2$)}
	\State $me.state \gets$ (\textsc{Stay})
	\If{$dist = 2$} \Comment{moving to the midpoint}
		\State $me.destination \gets other.position/2$
	\EndIf
\ElsIf{$other.state =$ (\textsc{Stay})}
	\State $me.state \gets$ (\textsc{Halted})
\algstore{myalg2}
\end{algorithmic}
\end{algorithm}

\begin{algorithm}
\begin{algorithmic}[1]
\algrestore{myalg2}
\Else \Comment{$other.state =$ (\textsc{Halted})}
	\If{$dist = 0$} \Comment{we have gathered}
		\State $me.state \gets$ (\textsc{Halted})
		\State terminate
	\Else \Comment{maintain position while I come}
		\State $me.state \gets$ (\textsc{Stay})
		\State $me.destination \gets other.position$
	\EndIf
\EndIf
\end{algorithmic}
\end{algorithm}

\begin{figure}[tbh]
\centering{\includegraphics[width=11cm]{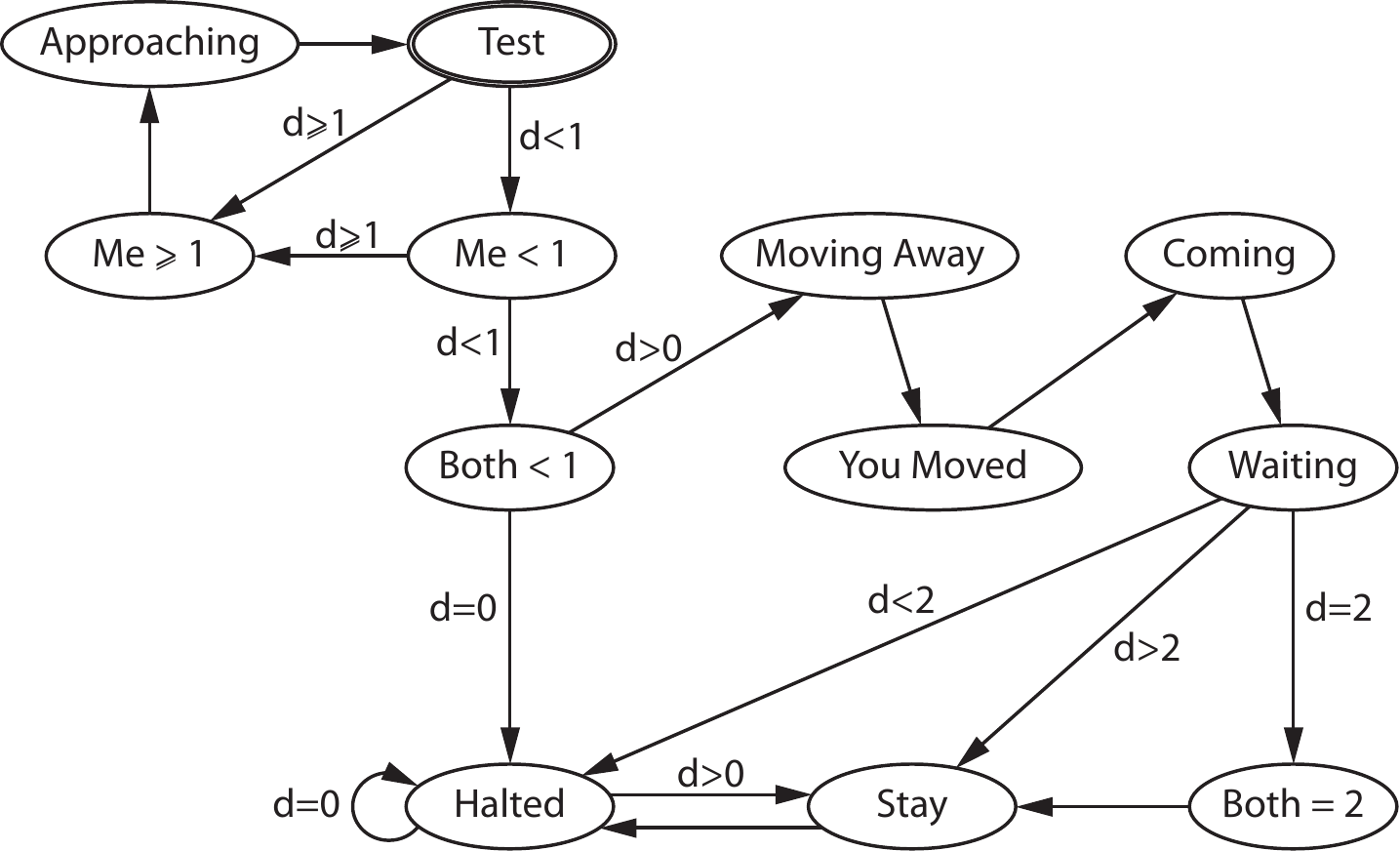}}
\caption{State transitions in Algorithm~\ref{alg3x}.}
\label{f1x}
\end{figure}

\begin{theorem}\label{th6x}
In  \asynch, \rend of two \FC robots is solvable with 12 colors. This result holds 
even without unit distance agreement.
\end{theorem}
\begin{proof}
We show that Algorithm~\ref{alg3x}, also depicted in Figure~\ref{f1x}, correctly solves \rend. Both robots start in state (\textsc{Test}), and then update their state to (\textsc{Me} $\geqslant 1$) or (\textsc{Me} $< 1$), depending if they see each other at distance greater or lower than $1$ (they may disagree, because their distance functions may be different).

If robot $r$ sees robot $s$ set to (\textsc{Me} $\geqslant 1$), it starts approaching it by moving to the midpoint, in order to reduce the distance. No matter if $r$ approaches $s$ several times before $s$ is activated, or both robots approach each other at different times, one of them eventually sees the other set to (\textsc{Approaching}). When this happens, their distance has reduced by at least a half, and at least one robot turns (\textsc{Test}) again, thus repeating the test on the distances.

At some point, both robots see each other at distance lower than $1$ during a test, and at least one of them turns (\textsc{Both} $<1$). If they have not gathered yet, they attempt to break symmetry by comparing their distance functions. To do so, when a robot sees the other set to (\textsc{Both} $<1$), it turns (\textsc{Moving Away}) and moves away by its own unit distance minus half their current distance. This move will be performed at most once by each robot, because if one robot sees the other robot still set to (\textsc{Both} $<1$), but it observes a distance not lower than $1$, then it knows that it has already moved away, and has to wait.

When a robot sees its partner set to (\textsc{Moving Away}), it shares this information by turning (\textsc{You Moved}). If only one robot turns (\textsc{You Moved}), while the other is still set to (\textsc{Moving Away}), then the second robot turns (\textsc{Coming}) and reaches the other robot, which just turns (\textsc{Waiting}) and stays still until they gather.

Otherwise, if both robots see each other set to (\textsc{You Moved}), they both turn (\textsc{Coming}) and switch positions. At least one of them then turns (\textsc{Waiting}). Now, if a robot sees its partner set to (\textsc{Waiting}) and they have not gathered yet, it knows that their current distance is the sum of their unit distances. If such distance is greater than $2$, then the robot knows that its partner's unit distance is bigger, and it moves toward it, while ordering it to stay still. Vice versa, if the distance observed is smaller than $2$, the observing robot stays still and orders its partner to come.

Finally, if the distance observed is exactly $2$, the observing robot knows that the two distance functions are equal, and turns (\textsc{Both} $=2$). In this case, a simple protocol allows them to meet. If a robot sees the other set to (\textsc{Both} $=2$) at distance $2$, it turns (\textsc{Stay}) and moves to the midpoint. If both robots do so, they eventually gather. Indeed, even if the first robot reaches the midpoint while the other is still set to (\textsc{Both} $=2$), it now sees its partner at distance $1$, and knows that it has to wait. On the other hand, whenever a robot sees its partner set to (\textsc{Stay}), it turns (\textsc{Halted}), which tells its partner to reach it. This guarantees gathering even if only one robot attempts to move to thee midpoint.
\end{proof}

\subsection{Semi-synchronous}
In \ssynch the situation is radically different from the \asynch case. In fact,   it is possible to
find a simple solution  in  $\mathcal L$ that uses the minimum number of colors possible, and operates
correctly  without unit distance agreement, starting from any arbitrary color configuration, and with interruptable movements  (see Algorithm \ref{alg2x}  and Figure \ref{f3x}). 

\begin{algorithm}
\caption{\ Rendezvous for non-rigid \ssynch with three externally visible states}\label{alg2x}
\begin{algorithmic}[1]
\If{$other.state = A$}
	\State $me.state \gets B$
	\State $me.destination \gets other.position/2$
\ElsIf{$other.state = B$}
	\State $me.state \gets C$
\Else \Comment{$other.state = C$}
	\State $me.state \gets A$
	\State $me.destination \gets other.position$
\EndIf
\end{algorithmic}
\end{algorithm}

\begin{figure}[ht]
\centering{\includegraphics[width=4cm]{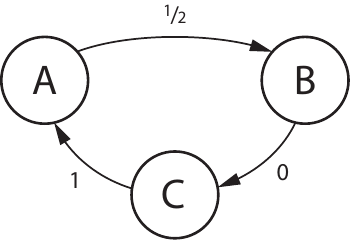}}
\caption{State transitions in Algorithm~\ref{alg2x}.}
\label{f3x}
\end{figure}

\begin{theorem}\label{th4x}
In  \ssynch, \rend  of two \FC robots is solvable by an algorithm in $\mathcal L$  with only three distinct colors.
This result holds even if  starting from an arbitrary color configuration, without unit distance agreement, and with   non-rigid movements. 
\end{theorem}
\begin{proof}
We show that Algorithm~\ref{alg2x} (see also Figure \ref{f3x}) correctly solves \rend from any initial configuration. Assume first that both robots start in the same state and both are activated at each turn. Then they keep having equal states, and they cycle through states $A$, $B$, and $C$ forever. Every time they are both set to $A$, they move toward the midpoint and their distance reduces by at least $2\delta$, until it becomes so small that they actually gather.

Otherwise, if at some point the two robots are in different states, they will keep staying in different states forever. In this case their distance will never increase, and they will periodically be found in states $B$ and $C$, respectively. Whenever this happens, the robot set to $C$ retains its state and stays still until the other robot is activated and moves toward it by at least $\delta$. As soon as their distance becomes not greater than $\delta$ and they turn again $B$ and $C$, they finally gather.
\end{proof}

Note that the number of colors used by the algorithm  is optimal. This follows as a corollary of the impossibility result 
when lights are visible to both robots:

\begin{lemma}\cite{viglietta}
In  \ssynch, \rend of two robots with persistent memory visible by both of them is unsolvable by algorithms in $\mathcal L$ that use only two colors.
\end{lemma}

\section{Movements: Knowledge vs.\ Rigidity}
In this section,  we consider the {\em Rendezvous} problem when the movement of the robots can be 
interrupted by an adversary; previously,  unless otherwise stated, we have considered rigid movements, i.e., in each cycle  a robot  reaches  its computed destination point.
Now, the only constraint on the adversary is that a robot, if
interrupted before reaching its destination, moves by at least $\delta>0$
(otherwise, rendezvous is clearly impossible).
We  prove that, for rendezvous with lights,  knowledge of $\delta$ has the same power as rigidity of the movements.
 Note that knowing $\delta$ implies also that the robots can agree on a unit distance.
 
\subsection{\FS Robots}
\begin{theorem}\label{th3x}
In non-rigid \ssynch,  \rend of two \FS robots with  knowledge of $\delta$  is solvable with three colors.
\end{theorem}
\begin{proof}
We show that Algorithm~\ref{alg1x} correctly solves \rend. Both robots start in state $A$. If a robot sees its partner at distance lower than $\delta/2$, it moves in the opposite direction, to the point at distance $\delta/2$ from its partner. On the other hand, if the distance observed is not lower than $\delta$, it moves toward the point located $\delta/4$ before the midpoint.

After sufficiently many turns, the robots are found at a distance in the interval $[\delta/2, \delta)$, and both in state $A$. From now on, all their movements are rigid. If only one robot is activated, it reaches its partner and \rend is solved. Otherwise, they both turn $B$ and switch positions. Then, if both robots are activated, they gather in the midpoint. Otherwise, one of them turns $C$ and moves to the midpoint. Now, the robot still in $B$ keeps staying still because it observes a distance lower than $\delta/2$. On the other hand, the robot set to $C$ moves to its partner as soon as it is activated.
\end{proof}

\begin{algorithm}
\caption{\ Rendezvous for non-rigid \ssynch with knowledge of $\delta$ and three internal states}\label{alg1x}
\begin{algorithmic}[1]
\State $dist \gets \lVert other.position\rVert$
\If{$dist = 0$}
	\State terminate
\EndIf
\If{$me.state = A$}
	\If{$dist < \delta/2$} \Comment{reach the point at distance $\delta/2$ from the other}
		\State $me.destination \gets other.position\cdot (1-\delta/(2\cdot dist))$
	\ElsIf{$\delta/2 \leqslant dist < \delta$} \Comment{gather or switch positions}
		\State $me.state \gets B$
		\State $me.destination \gets other.position$
	\Else \Comment{$dist \geqslant \delta$, reach the point at distance $\delta/4$ from the midpoint}
		\State $me.destination \gets other.position\cdot (\oneh-\delta/(4\cdot dist))$
	\EndIf
\ElsIf{$me.state = B$}
	\If{$\delta/2 \leqslant dist < \delta$}
		\State $me.state \gets C$
		\State $me.destination \gets other.position/2$
	\EndIf
\Else \Comment{$me.state = C$}
	\State $me.destination \gets other.position$
\EndIf
\end{algorithmic}
\end{algorithm}

\subsection{\FC Robots}
\begin{algorithm}
\caption{\ Rendezvous for non-rigid \asynch with knowledge of $\delta$ and three externally visible states}\label{alg4x}
\begin{algorithmic}[1]
\State $dist \gets \lVert other.position\rVert$
\If{$other.state =$ \textsc{Start}}
	\If{$dist = 0$} \Comment{we have already gathered}
		\State $me.state \gets$ \textsc{Come}
	\ElsIf{$dist \leqslant \delta$} \Comment{moving $\delta/2$ away}
		\State $me.state \gets$ \textsc{Start}
		\State $me.destination \gets - other.position \cdot \delta / (2\cdot dist)$
	\ElsIf{$dist > 2\delta$} \Comment{moving $\delta/2$ in}
		\State $me.state \gets$ \textsc{Start}
		\State $me.destination \gets other.position \cdot \delta / (2\cdot dist)$
	\Else \Comment{$\delta < dist \leqslant 2\delta$, ready to gather}
		\State $me.state \gets$ \textsc{Ready}
	\EndIf
\ElsIf{$other.state =$ \textsc{Ready}}
	\State $me.state \gets$ \textsc{Come}
	\If{$\delta < dist \leqslant 2\delta$} \Comment{reaching the midpoint}
		\State $me.destination \gets other.position/2$
	\EndIf
\Else \Comment{$other.state =$ \textsc{Come}}
	\If{$dist = 0$} \Comment{we have gathered}
		\State $me.state \gets$ \textsc{Come}
		\State terminate
	\Else
		\State $me.state \gets$ \textsc{Ready}
			\State $me.destination \gets other.position$
	\EndIf
\EndIf
\end{algorithmic}
\end{algorithm}
 
\begin{theorem}\label{th7x}
In non-rigid \asynch , \rend of two \FC robots with knowledge of $\delta$ is solvable with three colors.
\end{theorem}
\begin{proof}
We show that Algorithm~\ref{alg4x} correctly solves \rend. Both robots begin their execution in state \textsc{Start}, and attempt to position themselves at a distance in the interval $(\delta, 2\delta]$. To do so, they adjust their position by moving by $\delta/2$ at each step. When a robot sees its partner at the desired distance, it turns \textsc{Ready} and stops. It is easy to prove that, even if its partner is still moving, it will end its move at a distance in the interval $(\delta, 2\delta]$.

When a robot sees its partner set to \textsc{Ready}, it turns \textsc{Come} and moves to the midpoint. The midpoint is eventually reached, because the distance traveled is not greater than $\delta$. Assume that both robots turn \textsc{Ready}, both see each other, and move toward the midpoint. If robot $r$ reaches the midpoint and sees its partner $s$ still on its way and set to \textsc{Come}, $r$ turns \textsc{Ready} and keeps chasing $s$. When $s$ reaches its destination and sees $r$ set to \textsc{Ready} and at distance at most $\delta$, it stays still and waits until \rend is solved.

Similarly, assume that $r$ sees $s$ set to \textsc{Ready} and turns \textsc{Come} before $s$ sees $r$ set to \textsc{Ready}. $r$ will reach the midpoint and stay there, while $s$ will start chasing $r$ until they meet in the midpoint.
\end{proof}

\section{Open Problems}
We have shown that rendezvous can be obtained both in \FS and \FC, 
two models substantially weaker than the one of \cite{visbits}, where both internal memory and 
communication memory   capabilities are present. Our results open several new problems
and research questions.

 Our results, showing
that rendezvous is possible in \ssynch for \FS robots and in \asynch for \FC robots,
seem to indicate  that  ``it is better to communicate than to remember''.
However, determining the precise computational relationship between \FS and \FC is an open problem.
To settle it,  it must be determined whether or not  it is possible for \FS robots to rendezvous
in \asynch.

Although minimizing the amount of constant memory  was not the primary focus of this paper, the  number of states  
employed by our algorithms  is rather  small.  An interesting research question is to determine the smallest amount of
memory necessary for the robots to rendezvous when rendezvous is possible, and devise optimal solution protocols.

The knowledge of $\delta$ in non-rigid scenarios is quite  powerful and allows for simple solutions. It is an open problem to
study the {\em Rendezvous} problem  for \FS and \FC robots when $\delta$ is unknown or not known precisely.

This paper has extended the classical models of oblivious silent robots into
two directions: adding finite memory, and enabling finite communication. 
It  thus opens  the  investigation  in the \FS and \FC models
of  other  classical robots problems (e.g., {\em Pattern Formation}, {\em Flocking}, etc.);
an exception is  {\em Gathering} because, as mentioned in the introduction, it  is already solvable without
persistent memory and without communication \cite{gather}.

\end{document}